\documentclass[12pt]{amsart}  

\usepackage{graphicx}
\usepackage[bookmarks=false]{hyperref}
\usepackage{stmaryrd}
\usepackage{graphicx}
\usepackage{amsmath}
\usepackage{alltt}
\usepackage{mathrsfs} 
\usepackage{amsfonts}
\usepackage{amssymb}%
\usepackage{setspace}
\usepackage{amsthm}
\usepackage{amscd}
\usepackage{url}
\usepackage[all,knot,poly]{xypic} 
\usepackage{endnotes}
\usepackage[margin=1.5in]{geometry} 
\usepackage{tikz}
\usetikzlibrary{plotmarks}

\usetikzlibrary{decorations.markings}

\usepackage{enumerate} 
\usepackage{comment} 

\newcommand{\mcC}{\mathcal{C}}

\newcommand{\mcM}{\mathcal{M}}

\newcommand{\mcV}{\mathcal{V}}

\newcommand{\scrR}{\mathscr{R}} 
 
\newcommand{\scrT}{\mathscr{T}}

\newcommand{\R}{\mathbb{R}} 
\newcommand{\N}{\mathbb{N}} 

\newcommand{\ZZ}{\mathbb{Z}}

\newcommand{\ra}{\rightarrow}

\newcommand{\id}{\operatorname{id}}

\newcommand{\union}{\cup}

\newcommand{\Hom}{\operatorname{Hom}}

\newcommand{\Set}{{\rm Set}}

\newcommand{\NA}{\operatorname{\tiny NA}}

\theoremstyle{plain}
\newtheorem{theorem}{Theorem}[section]

\newtheorem{proposition}[theorem]{Proposition}

\newtheorem{observation}[theorem]{Observation}

\theoremstyle{definition}

\theoremstyle{definition}
\newtheorem{defn}[theorem]{Definition}
\theoremstyle{definition}

\theoremstyle{definition}
\newtheorem{example}[theorem]{Example}
\theoremstyle{definition}

\specialcomment{issue}
{\begingroup\color{red}\footnotesize}{\endgroup}

\excludecomment{baa}
\excludecomment{issue} 
\excludecomment{comment}

\tikzstyle{dom}=[rectangle,minimum height=20pt,minimum width=18pt,draw, fill=blue!20]
\tikzstyle{cod}=[minimum height=20pt,minimum width=18pt,draw, fill=red!20]

\title[Contextuality from missing and versioned data]{
  Contextuality from missing and versioned data}

\author[Jason Morton]{Jason  Morton${}^*$} 
\thanks{${}^*$Departments of Mathematics and Statistics, Pennsylvania State University. Supported by AFOSR Grant FA9550-16-1-0300}

\begin{document}
\begin{abstract} 
Traditionally categorical data analysis (e.g. generalized
linear models) works with simple, flat datasets akin to a single table
in a database with no notion of missing data or conflicting versions. In contrast, modern data analysis must deal with distributed
databases with many partial local tables that need not always agree. The computational agents tabulating these tables are spatially separated, with binding speed-of-light constraints and data arriving too rapidly for these distributed views ever to be fully informed and globally consistent.   Contextuality is a mathematical property which describes a kind of inconsistency arising in quantum mechanics (e.g. in Bell's theorem).  In this paper we show how contextuality can arise in common data collection scenarios, including missing data and versioning (as in low-latency distributed databases employing snapshot isolation). In the companion paper, we develop statistical models adapted to this regime.
\end{abstract}
\maketitle


\section{Introduction}

In this article we show that contextuality formally identical to the quantum weirdness of Bell's Theorem can arise when we analyze a database which is versioned or has missing data.  We call the increasingly common regime in which this can occur the {\em slow inconsistent regime}.

By {\bf slow}, we mean that analysis happens on the same timescale in which information is collected and transmitted.
By {\bf inconsistent}, we mean that we embrace the possibility that agents, views, factors, or marginal tables may have irreconcilably inconsistent views of the world.  This could occur because the data collection and analysis is done by multiple computational agents spatially distributed, who because analysis is slow never reach consensus.  It could occur because there is missing data, or because analysis is being done on snapshots that include stale records. 
By {\bf statistics}, we mean that we nevertheless want to do inference, make predictions and decisions, and fit models 
in this setting, ideally with quantified uncertainty and guarantees.

The input to a statistical analysis is typically a data frame such as Table \ref{table:perfectrows}.  Categorical data analysis (Agresti \cite{agresti2013categorical}) studies these data frames, or the contingency tables that summarize them.  

\begin{table}[h]
{\small
    \begin{tabular}{l l l l}
    {2014-09-17T00:00:30.750} & Red & Hi & Good \\
    {2014-09-17T00:00:31.152} & Red & Low & Good \\
    {2014-09-17T00:00:33.152} & Red & Hi & Good \\
    {2014-09-17T00:00:39.112} & Green & Low & Bad \\
    & $\vdots$ & &\\
  \end{tabular}\caption{\label{table:perfectrows}}
  }
\end{table}

A data frame might result from a select and is close to the notion of a table in a relational database. 
Traditional statistical analysis assumes that we have one perfect table in one perfect database.  A sometimes complex pipeline has combined all information into one table with no missing data, and rows which are timeless (i.i.d.) or exhibit linearly ordered time (as in a time series).  Such ideal data has been a favorite of statisticians for a century but is losing market share to more complex machine-generated data streams.

After summarizing the data in a contingency table, we might analyze with 
generalized linear models.  Workarounds when assumptions fail tend to involve ``perfecting'' the data first.  Disparate sources of data are matched and combined. An analyst might impute missing values or simply throw away incomplete rows.  This introduces bias, because there is no such thing as missing data -- only misspecified models that assume missingness is impossible, and that the pipeline perfecting the data does not affect the validity of the analysis.

A common estimate is that 80-90\% of the work of data analysis consists of cleaning and manipulating the data, joining tables, getting rid of missing data and so on.  At the end of this pipeline which distills a perfect table (e.g. a table of counts), statistical analysis is performed which treats the distillate as a simple draw from a straightforward distribution, bravely assuming that no bias or error was introduced in the long pipeline that constituted the majority of the work.  We aim to push statistics further back into the pipeline, by developing techniques that can be applied more directly to the data as it is collected and where it lies, with fewer constraints, assumptions and distillation steps.

In the slow, inconsistent regime, we often face structurally missing data, with missingness patterned by the distributed nature of the data processing activity.  Generally {\em every} row has missing data, because no agent or node is privy to the complete state of the universe.  Such omniscience is impossible given the finite speed of light and physical constraints on processing power. 
Thus we might replace the timestamps in Table \ref{table:perfectrows} with time intervals or (interval, location) spacetime eventstamps that represent our uncertainty, and allow for a pattern of NAs in the rows.


The remainder of this article is organized as follows.  We begin with two easy-to-understand motivating examples in Sections \ref{sec:missing2contextual} and \ref{sec:snapshot2contextual}, and then explain the mathematical details that describe what is happening in these examples and permit precise definitions. The examples show how contextuality identical to that in Bell's Theorem  can easily arise without quantum effects,  contrary to the assumptions of much of classical statistics and data analysis. In reading the examples, the reader may want to refer forward for the occasional definition.  

First, in Section \ref{sec:missing2contextual} we show how contextual models can arise solely from missing data in a data table or a database.  Thus we can expect contextuality to appear when there is missing data, and rows with missing data are skipped in computing summaries. 
Second, in Section \ref{sec:snapshot2contextual} we show how contextuality can arise from write skew in a database using snapshot isolation, even without missing data and under various consistency constraints.  Thus we can expect contextuality to appear in modern distributed systems, especially those under time pressure to perform writes and reads without waiting for global consensus.

Next we detail our data model.  Section \ref{sec:time} sketches a simplified model of spacetime; it features partially ordered events and an interval time model. Interval time is implemented in existing distributed database systems, but requires strong guarantees bounding a clock's divergence from a common reference clock. Our model can serve as a practical model for the causal structure underlying the version poset used in Sections \ref{sec:snapshot2contextual} and \ref{sec_contextual}.


Section \ref{sec_formaltables} develops our single-table data model, which allows for tables (with indexed data items) rather than relations, missing data, event-versioned items, and prepares us for measurement contexts.

In Section \ref{sec_contextual} we build on the sheaf-theoretic view of Abramsky and coauthors \cite{abramsky2011sheaf,abramsky2011cohomology,abramsky2013relational,abramsky2015contextuality,abramsky2016quantifying}, using presehaves of tables or sections of a presheaf to model a database with many tables, subject to some consistency constraints. This provides a definition of contextuality in terms of the absence of a global section.  Note that contextuality has appeared in other settings including cognitive science \cite{narens2014alternative}.

Having established the prevalence of contextual models, in the companion paper we address three related problems this suggests: how to quantify contextuality, how to extend statistical models to the slow inconsistent regime, and how to fit these models in the presence of contextuality. 




\section{Contextual models from missing data} \label{sec:missing2contextual}







Let $A,B,A',B'$ be binary variables, and consider the following set of four marginal $\ZZ$-relations (contingency tables), chosen because their empirical distributions (1/8 times each table) give the Bell family.

$m_{AB} = \bordermatrix{~ & b_0 & b_1 \cr
              a_0 & 4 & 0 \cr
              a_1 & 0 & 4 \cr}$,
$m_{A'B} = \bordermatrix{~ & b_0 & b_1 \cr
              a'_0 & 3 & 1 \cr
              a'_1 & 1 & 3 \cr}$,
$m_{AB'} = \bordermatrix{~ & b'_0 & b'_1 \cr
              a_0 & 3 & 1 \cr
              a_1 & 1 & 3 \cr}$, and 

$m_{A'B'} = \bordermatrix{~ & b'_0 & b'_1 \cr
              a'_0 & 1 & 3 \cr
              a'_1 & 3 & 1 \cr}.$

Note that these $m$ are a compatible family of local sections (Def.\ \ref{def_compat_family}) for a presheaf of $\mathbb{Z}$-relation-spaces (Def. \ref{def_presheaf_relation_spaces}); the intersections are all single-variable with marginal $(4,4)$. Could these glue to a global relation, or be produced by summary maps $\pi_{AB}$, $\pi_{A'B}$, $\pi_{AB'}$, and $\pi_{A'B'}$?   

There exists no global table without missing data that $\pi$-projects to these marginal tables.  Because the records must be complete (have a definite value for each of the four variables), there must be exactly eight rows.  In four rows, we must have $A\!=\!0$ and $B\!=\!0$, i.e.\ the table must look like this (with * in unknowns):

\begin{center}
\begin{tabular}{c|cccc}
  index & A & B & A' & B' \\
  \hline
  1   & 0 & 0 & * & * \\
  2   & 0 & 0 & * & * \\
  3   & 0 & 0 & * & * \\
  4   & 0 & 0 & * & * \\
  5   & 1 & 1 & * & * \\
  6   & 1 & 1 & * & * \\
  7   & 1 & 1 & * & * \\
  8   & 1 & 1 & * & * \\
\end{tabular}.
\end{center}

Turning now to marginal table $m_{AB'}$, in three of the cases that $A=0$, we must have that $B'=0$, and in one that $B'=1$.  Similarly in three of the cases that $A=1$, we must have that $B'=1$, and in one that $B'=0$.  We can assign these freely:

\begin{center}
\begin{tabular}{c|cccc}
  index & A & B & A' & B' \\
  \hline
  1   & 0 & 0 & $x_1$ & 0 \\
  2   & 0 & 0 & $x_2$ & 0 \\
  3   & 0 & 0 & $x_3$ & 0 \\
  4   & 0 & 0 & $x_4$ & 1 \\
  5   & 1 & 1 & $x_5$ & 1 \\
  6   & 1 & 1 & $x_6$ & 1 \\
  7   & 1 & 1 & $x_7$ & 1 \\
  8   & 1 & 1 & $x_8$ & 0 \\
\end{tabular}.
\end{center}

Considering the marginal table $m_{A'B}$, in three of the cases that $B=0$, we must have that $A'=0$, and in one that $A'=1$.  Similarly in three of the cases that $B=1$, we must have that $A'=1$, and in one that $A'=0$.  Then
\[
(C1)\; x_1 + x_2 + x_3 + x_4 = 1\;\; \text{and} \;(C2)\; x_5 + x_6 + x_7 + x_8 = 3.
\]

On the other hand considering $m_{A'B'}$, in three of the cases that $B'=0$, we must have that $A'=1$, and in one that $A'=0$.  Similarly in three of the cases that $B'=1$, we must have that $A'=0$, and in one that $A'=1$.  Then
\[
(C3)\; x_1 + x_2 + x_3 + x_8 = 3\;\; \text{and} \;(C4)\; x_4 + x_5 + x_6 + x_7 = 1.
\]
Then subtracting  $(C3) - (C1)$ we have $x_8 - x_4 = 2$, which is impossible because $x_4$ and $x_8$ are both either $0$ or $1$.

We can also see this by passing to tables of empirical marginal probabilities, and applying Bell's inequalities.

However consider the enlarged Table \ref{table:BellNA} of records that includes missing data.

\begin{table}[h]
\begin{tabular}{c|cccc}
  index & A & B & A' & B' \\
  \hline
  1   & 0 & 0 & 0 & 0 \\
  2   & 1 & 1 & 1 & 1 \\
  3   & 0 & 0 & 0 & \tiny{NA} \\
  4   & 1 & 1 & \tiny{NA} & 1 \\
  5   & \tiny{NA} & 0 & 1 & 0 \\
  6   & 1 & \tiny{NA} & 1 & 0 \\
  7   & \tiny{NA} & \tiny{NA} & 1 & 0 \\
  8   & \tiny{NA} & 1 & 0 & 1 \\
  9   & 0 & \tiny{NA} & 0 & 1 \\
  10  & \tiny{NA} & \tiny{NA} & 0 & 1 \\
  11  & 1 & 1 & 1 & \tiny{NA} \\
  12  & 0 & 0 & \tiny{NA} & 0 \\
  13  & 0 & 0 & 0 & \tiny{NA} \\
  14  & 1 & 1 & 1 & \tiny{NA} \\
  15  & 0 & \tiny{NA} & \tiny{NA} & 0 \\
  16  & 1 & \tiny{NA} & \tiny{NA} & 1 \\
\end{tabular}\caption{ A table with missing data. Applying available-case analysis to the table to estimate the distributions in contexts $C_{AB},C_{AB'}, C_{A'B}, C_{A'B'}$ yields a Bell family.
  \label{table:BellNA}}
\end{table}

A typical method of dealing with missing data is simply to throw it out (available-case analysis, see Observation \ref{obs_avail}).  In other words, to compute the marginal table $m_{AB}$, we use the $\pi_{AB} \circ \tau_{AB}$ map which first restricts to rows in which both $A$ and $B$ are not $NA$, and then sums these to produce the summary table.  Applying this procedure to Table \ref{table:BellNA}, we obtain exactly the desired marginals, $m_{AB}$, $m_{A'B}$, $m_{AB'}$, and $m_{A'B'}$ (although in this example the single-variable marginals only have correct proportions).  Of course the same is true if we ask for proportions.  Thus we have shown the following.

\begin{proposition}
Categorical data with missing data can result in inconsistent marginal counts and proportions identical to those that arise from quantum nonlocality.
\end{proposition}

Consequently, if we consider a model that is fit to such data using only summary data from marginal tables (sufficient statistics), the possibility arises that there is no global joint.  In the sequel we develop a generalized notion of exponential family model that is adapted to this scenario.

\section{Contextual models from write skew, snapshot isolation and multiversion concurrency control}  
\label{sec:snapshot2contextual}








\begin{example} \label{example_versionbell}

  Consider a family of local sections with contexts $C_{AB}$, $C_{AB'}$, $C_{A'B}$, $C_{A'B'}$ produced by the versioned Table \ref{tab_versionsbell} treated as a global section, with $\omega(C_{AB})=5$, $\omega(C_{A'B})=2$, $\omega(C_{AB'})=3$, and $\omega(C_{A'B'})=4$.

  \begin{table}[hb]
\begin{tabular}{cc|cccc}
  version & index & A & B & A' & B' \\
  \hline
  1 & 1   & 0 & 0 & 0 & 0 \\
  1 & 2   & 0 & 0 & {\bf 0} & 0 \\
  1 & 3   & 0 & {\bf 1} & 0 & {\bf 1} \\
  1 & 4   & {\bf 0} & 1 & 0 & 1 \\
  1 & 5   & {\bf 1} & 0 & 1 & 0 \\
  1 & 6   & 1 & {\bf 0} & 1 & {\bf 0} \\
  1 & 7   & 1 & 1 & {\bf 1} & 1 \\
  1 & 8   & 1 & 1 & 1 & 1 \\
  \hline
  2 & 3   & 0 & {\bf 0} & 0 & 1 \\
  2 & 6   & 1 & {\bf 1} & 1 & 0 \\
  \hline
  3 & 3   & 0 & 1 & 0 & {\bf 0} \\
  3 & 6   & 1 & 0 & 1 & {\bf 1} \\
  \hline
  4 & 2   & 0 & 0 & {\bf 1} & 0 \\
  4 & 7   & 1 & 1 & {\bf 0} & 1 \\
  \hline
  5 & 4   & {\bf 1} & 1 & 0 & 1 \\
  5 & 5   & {\bf 0} & 0 & 1 & 0 \\
\end{tabular}
\caption{A table with versioned edits. In this eight-record, five version table, the version partial order is the tree $(1 (2(5))(3)(4))$. From a common version 1, edits 2, 3, and 4 are concurrent.  Edit 2 swaps the state of variable $B$ in records 3 and 6, while edit 3 swaps the state of variable $B'$ in records 3 and 6; edit 4 swaps the state of variable $A'$ in records 2 and 7; and edit 5 follows edit 2, swapping the states of variable $A$ in records 4 and 5. Edited states are bolded.  Contexts are assigned versions $\omega(C_{AB})=5$, $\omega(C_{A'B})=2$, $\omega(C_{AB'})=3$, and $\omega(C_{A'B'})=4$. From the point of view of these four observers (see Table \ref{tab_version_sharded}), the system is a Bell family.  \label{tab_versionsbell}}
  \end{table}

  Each edit maintains the invariant that the marginal count of any single variable has four zeros and four ones by swapping two variables, so the compatibility condition is maintained at the level of counts.

The concurrent persepctive of four agents, one responsible for each context, can be described by a $\pi$-compatible concurrent snapshot $T^\omega$ (Def. \ref{def_pi_compcon_snap}).  The four tables $T^\omega(\{A,B\}) = T_{\leq E_5}(AB)$, $T^\omega(\{A,B'\})$, $T^\omega(\{A',B\})$, and $T^\omega(\{A',B'\})$ are given in in Table \ref{tab_version_sharded}.  Note that $T^\omega$ gives a compatible family of local sections of a presheaf of table spaces (so also compatible presheaves of tables), because conflicts are resolved by version numbers.  Forgetting the version numbers, we get disagreement on indexed overlaps; forgetting indices as well, we recover agreement on overlaps.

  For example variable $A$, index 4 in $T^\omega(\{A,B\})$ has $\sigma_A = 1$ while in $T^\omega\{A,B'\}$ it has $\sigma_A = 0$ (but the version numbers are different). As an unindexed multiset, or eqivalently passing to the summary $\N$-relations, we obtain again a compatible family (this was enforced by only using swap operations that maintain the invariant).

  Each snapshot is simply an $8$-record data table with no missing data.  Nevertheless, $\frac{1}{8}\pi_\phi T^\omega$ is the Bell family. Summarizing $T_{\leq E_5}(AB)$, $T_{\leq E_2}(A'B)$, $T_{\leq E_3}(AB')$, and $T_{\leq E_4}(A'B')$ with $\pi_{AB}$, $\pi_{AB}$, $\pi_{AB}$, and $\pi_{AB}$ to obtain relations, we obtain
  
 $R_{\leq E_5}(AB) = \bordermatrix{~ & b_0 & b_1 \cr
              a_0 & 4 & 0 \cr
              a_1 & 0 & 4 \cr}$,  
$R_{\leq E_2}(A'B) = \bordermatrix{~ & b_0 & b_1 \cr
              a'_0 & 3 & 1 \cr
              a'_1 & 1 & 3 \cr}$
  
$R_{\leq E_3}(AB') = \bordermatrix{~ & b'_0 & b'_1 \cr
              a_0 & 3 & 1 \cr
              a_1 & 1 & 3 \cr}$, and 
$R_{\leq E_4}(A'B') = \bordermatrix{~ & b'_0 & b'_1 \cr
              a'_0 & 1 & 3 \cr
              a'_1 & 3 & 1 \cr}.$

    \begin{table}[htb]
\begin{tabular}{cc|cccc}
  version & index & A & B\\
  \hline
  1 & 1   & 0 & 0\\
  1 & 2   & 0 & 0\\
  2 & 3   & 0 & {\bf 0}\\
  5 & 4   & {\bf 1} & 1\\
  5 & 5   & {\bf 0} & 0\\
  2 & 6   & 1 & {\bf 1}\\
  1 & 7   & 1 & 1\\
  1 & 8   & 1 & 1\\
\end{tabular}$\qquad$
\begin{tabular}{cc|cccc}
  version & index & B & A'\\
  \hline
  1 & 1   & 0 & 0\\
  1 & 2   & 0 & {\bf 0}\\
  2 & 3   & {\bf 0} & 0\\
  1 & 4   & 1 & 0\\
  1 & 5   & 0 & 1\\
  2 & 6   & {\bf 1} & 1\\
  1 & 7   & 1 & {\bf 1}\\
  1 & 8   & 1 & 1\\
\end{tabular}\\
\begin{tabular}{cc|cccc}
  version & index & A & B'\\
  \hline
  1 & 1   & 0 & 0 \\
  1 & 2   & 0 & 0 \\
  3 & 3   & 0 & {\bf 0} \\
  1 & 4   & {\bf 0} & 1 \\
  1 & 5   & {\bf 1} & 0 \\
  3 & 6   & 1 &  {\bf 1} \\
  1 & 7   & 1 & 1 \\
  1 & 8   & 1 & 1 \\
  \hline
\end{tabular}$\qquad$
\begin{tabular}{cc|cccc}
  version & index & A' & B'\\
  \hline
  1 & 1   & 0 & 0 \\
  4 & 2   & {\bf 1} & 0 \\
  1 & 3   & 0 & {\bf 1} \\
  1 & 4   & 0 & 1 \\
  1 & 5   & 1 & 0 \\
  1 & 6   & 1 & {\bf 0} \\
  4 & 7   & {\bf 0} & 1 \\
  1 & 8   & 1 & 1 \\
\end{tabular}

\caption{The snapshots $T^\omega(\{A,B\}) = T_{\leq E_5}(AB)$, $T^\omega(\{A',B\}) = T_{\leq E_2}(A'B)$, $T^\omega(\{A',B\}) = T_{\leq E_5}(A'B)$ and $T^\omega(\{A',B'\}) = T_{\leq E_2}(A'B')$ taken from Table \ref{tab_versionsbell}. \label{tab_version_sharded}}
  \end{table}

\end{example}


\section{Partially ordered events in a graph-interval space-time} \label{sec:time}


Despite being the fastest thing in the universe, to a processor core light is slow: it travels only about one foot per nanosecond.  In that time a computer might complete four cycles with each of its many cores.  In any real system, information usually travels slower still, as a beam of light bounces down a fiber, signals are packetized and attenuated, latency creeps in throughout the application stack, or a human is in the loop.

Traditional statistics happens on data that is generated and collected very slowly and carefully relative to the analysis process.  Real-time statistics on data generated by networked computers is more like a person trying to hand-analyze data being collected on various star systems spread throughout the local galactic neighborhood.  Each agent never really knows what is going on; by the time it does, it might not be relevant anymore; and every planet had a different, wildly out-of-date, view of the universe's state.

One useful idea for formalizing this situation is the {\em causet}.  A partial order $X$  is {\em locally finite} if for all $x,y \in X$ the cardinality of $\{x : x \leq z \leq y\}$ is finite.  A Lorentzian manifold is {\em weakly causal} if it contains no closed timelike curves.  Then a {\em causet} is a locally finite partial order which can be embedded in a weakly causal Lorentzian manifold.

In such a partial order, the points are spacetime events.  The relation $E_1 \leq E_2$ means that  $E_1$ is to the past of $E_2$ and so can cause or predict it.  Non-comparable events in the order are ``spacelike,'' and are simultaneous in some reference frame.


Unfortunately, it gets worse.  The partial order on events in a networked system is unknown and unknowable on a small timescale because the clocks themselves are inaccurate.  What we actually are faced with is a kind of blurry causet, with only a probability of one event $E_1$ being to the past of $E_2$.  This probability approaches one quickly as the delay between the events increases, but it is a serious issue inside the error of the clocks.


The traditional idea of a ``time series'' requires looking only at one point in space (varying over time) and/or making unwarranted assumptions about total ordering.  A more appropriate time series model in this context will look at the (likely truncated) past lightcone $\{F: F \leq E\}$  of an event $E$.  

\subsection{Interval model of time; resulting partial order}
We adopt the Corbett et al. \cite{corbett2013spanner} approach to partially ordered time, called TrueTime at Google.  This is to assume that clocks have some divergence $\epsilon$ from UTC (a shared global inertial frame) and treat each event timestamp as an interval $[t_s, t_e]$ such that the true UTC time $t_{UTC}$ is guaranteed to be in the interval.  While not a perfect representation of uncertainty (the guarantee could fail, and we might have more information about the probability density of the location of the true time inside the interval), it represents a good compromise.

We assume that the system under consideration involves finitely many interacting agents.  By definition, time for each agent is linearly ordered; an agent is something like a computational core or Turing machine which must process events in a single stream.  Each agent can associate an interval timestamp to any event with the above guarantees.  Our events are approximations to points; if an event in the colloquial sense lasts say 10 seconds, it should be broken into two events in our sense, one $E_1$ for the beginning of the long-running event and another $E_2$ for the end.

We also assume that agents are separated in space, so that classical information requires time to travel from one agent to another, bounded by the speed of light.

Thus if Alice assigns stamp $[t_1,t_2]$ to event $E$ and Bob assigns stamp $[t_3,t_4]$ to event $F$, $t_2 < t_3$ implies $E$ happens before $F$ in the UTC frame.  If further the interval $t_3 -t_2$ exceeds the time required for light to travel from Alice to Bob, we say that $E$ is causal to $F$, or that $F$ lies in the future light cone of $E$.  This provides a partial order of strict causality.

If we only have that $t_1 < t_4$, we say that $E$ may have occurred before $F$ in the UTC frame, but if $t_3 \leq t_2$, the reverse may also be true.  If $t_4-t_1$  exceeds the time $d_{A,B}$ required for light to travel from Alice to Bob, we say that $E$ is possibly causal to $F$, or that $F$ might lie in the future light cone of $E$.  If $t_2 - t_3 < d_{A,B}$, the reverse relations also hold.

We can arrange the information travel times (in a general system, they may not be symmetric) in a matrix $d$ of nonnegative reals or a graph $\Gamma$ of communication link delays, where $d$ is obtained by all shortest paths.  Together with interval timestamps, this induces a strict causality partial order on events, in which $E < F$ means that we are certain there was enough time between the latest possible end of $E$ and the earliest possible start of $F$ for information to travel from the agent recording $E$ to the agent recording $F$.

The basic unit of analysis is a {\em patch} of spacetime (See Figure \ref{fig_patch_hourglass}), consisting of the data of a connection graph $\Gamma$ and travel time matrix $d_\Gamma$, a set of {\em eventstamps} which are interval-timestamped events for each node in the graph, and the resulting partial orders.  

\begin{figure}
  \includegraphics[clip,trim=5cm 15cm 5cm 15cm, width=5cm]{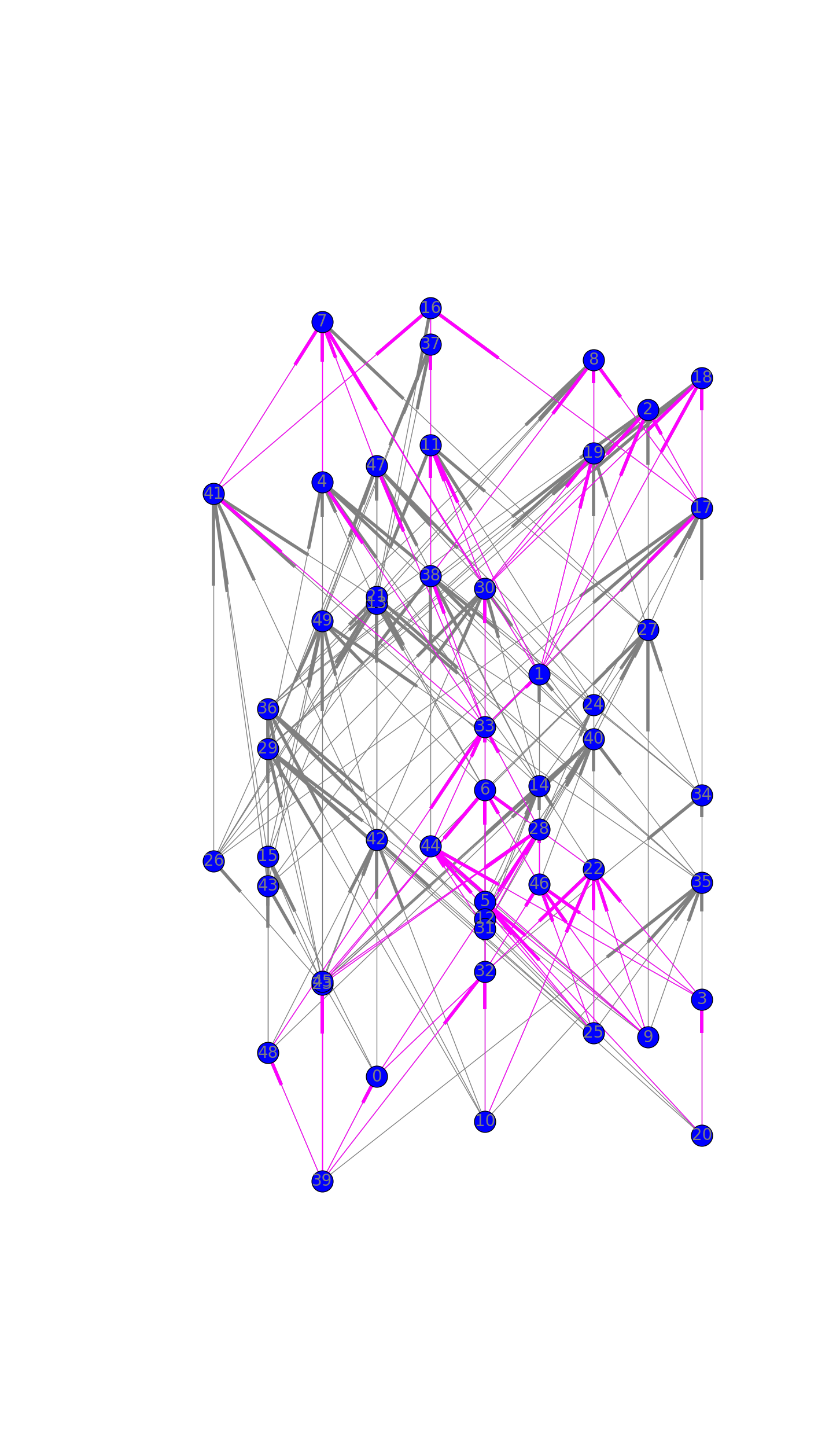}
  \caption{A small patch with a past and future lightcone of an event highlighted. \label{fig_patch_hourglass}}
\end{figure}

The eventstamps $v \in V$ in such a partial order will be used as version labels for the rows of our database tables.

\section{Tables, relations, restrictions, and summaries} \label{sec_formaltables}

For clarity, let us consider a common generalization of (1) the {\em measurement scenario} of quantum information and (2) the {\em hierarchical} or {\em loglinear} model of probabilistic graphical models.  We begin with variables (object variables) $X_1, \dots, X_n$, where $X_i$ has finite state space $\Sigma_i$.  For a subset $S \subset [n]$ define $\Sigma_S = \prod_{i \in S} \Sigma_i$.  We have {\em contexts}  $C_k \subset [n]$ or $\subset \{X_1, \dots, X_n\}$, grouped in a {\em measurement scenario} $\mcM = \{C_1, \dots C_K\}$; each context is a collection of variables. Conceptually, contexts are collections of variables which can be simultaneously observed.  For example, contexts might correspond to the maximal tables without stale or missing data, with rows assignments $\sigma_C \in \Sigma_C$ of $X_C$.  They also represent a cover of $[n]$.  We require contexts to be coatoms in the partial order by inclusion (so no context is a subset of another).  We also consider a topology on $[n]$, usually the discrete topology or the topology generated by the contexts.

For modeling, we further have factors $f_1, \dots, f_F$ with $f_j \subset \{X_1, \dots, X_n\}$. For now, we assume each factor is a subset of at least one context.  With this description, hierarchical models are the special case where there is only one context containing all the variables, while the measurement scenarios studied by Abramsky et al.\ and others are the special case where the factors and contexts exactly coincide (the model is saturated in each context).  
We will return to factors when we discuss models and model fitting.

\subsection{Relations, tables, versioned tables, and missing data}
Let $\mathbb{S}$ be a semiring. A function $R\!:\Sigma_S \ra \mathbb{S}$ is a semiring-valued relation.  For example $C\!:\Sigma_S \ra \mathbb{N}$ is called a {\em contingency table} \cite{agresti2013categorical} and might tabulate how many times each state $\sigma_S \in \Sigma_S$ was observed in a data set.  Similarly $P\!:\Sigma_S \ra \mathbb{R}_{\geq 0}$ with $\sum_{\sigma_S} P(\sigma_s )=1$ (e.g.\ if $P = C/\sum_{\sigma_S} C(\sigma_s )$) is a probability distribution. An element $\sigma_S \in \Sigma_S$ is a state and $(\sigma_S,s)$ with $s \in \mathbb{S}$ is a tuple in a relation on variables or ``columns'' $S\subset [n]$.

Note that we are identifying absence of the state $\sigma_S$ in the relation with the condition $R(\sigma_S)=0$; whether the relation is represented sparsely as a list of the states $\sigma_S$ with nonzero values, or densely as a function $R\!:\Sigma_S \ra \mathbb{S}$ is not specified.

For any partition $T \cup U = S$ with $T \cap U = \emptyset$ we can also view a $R\!:\Sigma_S \ra \mathbb{S}$ as a semiring-valued relation $R: \Sigma_T \times \Sigma_U \ra \mathbb{S}$ or a morphism $R\!:T \ra U$ (see e.g. \cite{morton2013generalized} for an application), a perspective useful for joins and other operations.  This is a useful and common category; for example the (monoidal) well-supported \cite{rosebrugh2005generic} compact closed category of complex-valued relations is equivalent to the well-supported compact closed category of complex finite-dimensional vector spaces and linear transformations, in which each vector space is equipped with an orthonormal basis.  Letting $\mathbb{S}$ be the Boolean semiring we obtain the category of sets and relations.

Alone, such a relation is not a good model of a database table, since we may require more than one copy of each state.  This is corrected by adding a unique index or primary key set to the relation.   
A semiring-valued relation $R\!:\Sigma_S \ra \mathbb{S}$ is also a functional (right-unique) relation in $\Sigma_s \times \mathbb{S}$: it identifies absence of a tuple with a zero assignment, and every state has exactly one semiring element.
In defining a table, we can maintain the surjective projection to $\Sigma_S$, which requires assigning an index even to $\sigma_S$ with value 0, or we can not.
This is a sparse vs.\ dense representation implementation issue, so for convenience we do assign such indices although there is a bit of awkwardness. 
That is, we have a relation $T \subset I \times \Sigma_S \times \mathbb{S}$ such that the projection $\pi_I\:T \ra I$ is injective (no two tuples in $T$ share the same index) and the projection to $\Sigma_S$ is surjective (every state has at least one index, at least conceptually). This implies that if $(i,\sigma_S, s)$ and  $(i,\sigma'_S, s')$ are both in $T$, $\sigma_S = \sigma'_S$ and $s=s'$. 


Versioning (as in the widely implemented multiversion concurrency control \cite{reed1978naming,bernstein1981concurrency} and snapshot isolation \cite{adya1999weak,cahill2009serializable,yabandeh2012critique,kulkarni2014logical}) is a useful tool for handling concurrent data reads and writes.  A data item represented by a pair $(i,\sigma_S)$ is allowed to have have multiple versions arranged in partial order which we require to be a causet. 
Let $V$ be a poset of versions (in particular, eventstamps from events in a causet patch as in Section \ref{sec:time}).

To define a {\em versioned table with columns $S$, index $I$, and version set $V$}, we extend the data table further to $T \subset V \times I \times \Sigma_S \times \mathbb{S}$.  We require instead of $\pi_I$ being injective that the projection $\pi_{V \times I}\!:T \ra V \times I$ is injective: no two tuples in $T$ share the same version and index.  Equivalently we are requiring that each transaction or commit $T_v$ is a table in the former sense; hence $T_v$ is a {\em change} to a table, with the snapshot $T_{\leq v}$ defined below being the changes at $v$ applied to the table state at the beginning of the transaction at $v$ . This implies that if $(v,i,\sigma_S, s)$ and  $(v,i,\sigma'_S, s')$ are both in $T$, $\sigma_S = \sigma'_S$ and $s=s'$. 

A {\em snapshot} $T_{\leq v}$ of a table at version $v$ contains all the $(u,i,\sigma_S,s)$ such that $u \leq v$ and $u$ is the maximal version such that $(i,\sigma_S,s) \in T_u$, i.e. it contains the latest version of all data items in the inclusive past of $v$ (the past lightcone/order ideal).

Finally let us add in the possibility of missing data: the state of any variable in any observation can be unknown, indeterminant, or uncollected.  One approach is to augment each variable's state space $\Sigma_j, j \in [n]$ with an $\NA$ symbol, setting $\tilde{\Sigma}_j \!=\! \Sigma_j \union \{\NA\}$.

These $\NA$s allow for arbitrary extensions of local state, can sometimes be cleaned by restriction, and are usually ignored or imputed in summaries.  
Let $S \subset U \subset [n]$. 
Define an extension map $\epsilon_{S \ra U}: \tilde{\Sigma}_S \ra \tilde{\Sigma}_U$ (also defined on $\Sigma_S$) by putting $\NA$ in all the slots in $U \setminus S$. Then we extend from (versioned) tables with columns $S$ to (versioned) tables with columns $U$ by sending $(i,\sigma_S, s)$ to $(i, \epsilon_{S \ra U} (\sigma_S), s)$ (versioned, $(v,i,\sigma_S, s)$ to $(v,i, \epsilon_{S \ra U} (\sigma_S), s)$).

We also have a restriction map of sets $\tau_{S \leftarrow U}\!: \tilde{\Sigma}_U \ra \tilde{\Sigma}_S$ which ignores the state of the variables we are not interested in.  Thus $\tau_{S \leftarrow U} \circ \epsilon_{S \ra U}$ is the identity map of sets on $\Sigma_S$, while $\epsilon_{S \ra U} \circ \tau_{S \leftarrow U}$ replaces the $U \setminus S$ states in a state $\sigma_U$ with $\NA$s.

Alternatively once can work without an $\NA$ symbol by talking only about extensions and restrictions, or using a mask, but these are typically more awkward. 

\subsection{Restrictions, extensions, and summaries of tables}
Let $S \subset U \subset [n]$,  and let $T$ and $T'$ be versioned tables with missing data, on columns $S$ and $U$ respectively. Define extension in terms of the operation on the state sets, 
\[
\epsilon_{S \ra U}(T) (v,i,\sigma_S,s) = T (v,i, \epsilon_{S \ra U} \sigma_S,s) 
\]
and restriction similiarly,
\begin{equation} \label{eq_tablerestriction}
\tau_{S \leftarrow U}(T') (v,i,\sigma_U,s) = T' (v,i, \tau_{S \leftarrow U} \sigma_U,s).
\end{equation}
Note that these operations do not affect the set of versions or indices in use, but could result in degenerate situations of various sorts such an index corresponding to a row which is entirely $\NA$.


A relation is distinguished from a table by having no index, so a unique semiring element for each state $\sigma_S$.  A {\em summary} of a table $T$ on columns $S$ produces a relation on $S$ from a table by summing over indices with the same state in the semiring, possibly changing semirings before summation. Thus a marginal relation of counts is obtained by combining a restriction with a summary.

Let $T$ be a table with missing data on columns $S$, and let  $\phi\!: \mathbb{S} \ra \mathbb{S}'$. Usually $\phi$ is an identity on $N$ or $\R$, $\mathbb{S} = \mathbb{B}$ and  $\mathbb{S}' = \mathbb{N}$ with $\phi({\rm True})=1$ and $\phi({\rm False})=0$, or $\phi: \N \ra \R_{\geq 0}$ divides by the total count to obtain probabilities.

Define the summary as a relation on $\Sigma_S$ by 
\begin{equation} \label{eq_summary_table2relation}
\pi_\phi(T)(\sigma_S) = \sum_{(i,\sigma'_S,s) \in T: \sigma'_S = \sigma_S} \phi \circ T(i,\sigma'_S).
\end{equation}

Since $\sigma_S \in \Sigma_S$,  $\sigma'_s \neq \sigma_S$ if $\sigma_S$ contains any $\NA$s, so that this summation operator $\pi$ {\bf skips any row with at least one $\NA$} in the specified columns $S$. The resulting $\pi_\phi(T)$ is a relation with no missing data.

\begin{observation}
This skip-$\NA$ method of dealing with missing data is the default for most analysts and software packages, although arguably not the best (see the survey \cite{allison2002missing}).  It is also called listwise deletion or complete case analysis.  Under the missing completely at random assumption (MCAR) in a single table it will not bias parameter estimates, because it is equivalent to taking a random subsample.  Allison \cite{allison2002missing} also observes that it is suprisingly robust to violations of MCAR, especially when what is missing are predictor variables in virtually any kind of regression.
\end{observation}

The summary $\pi_\phi$ is also defined on any commit $T_v$ or snapshot $T_{\leq v}$ by simply ignoring the version information (since in either case there is at most one version corresponding to any index).

Note that given $S \subset U \subset [n]$, and a table on $U$, we can obtain a marginal relation of counts on $S$ by either $\pi_\phi \circ \tau_{S \leftarrow U} T$ or $\tau_{S \leftarrow U} \circ \pi_\phi T$.  In general these are unequal if $T$ has missing data, because the summary produced by applying $\pi$ first might discard rows whose restriction  $\tau_{S \leftarrow U}$ has no missing data, so that those rows would be counted if restriction were applied before summarization.

\begin{proposition}
Without missing data, restriction and summarization commute; with missing data they do not.
\end{proposition}

Note that when $\mathbb{S}$ is a field and with the obvious vector space structure on tables, restriction and summarization are linear, and each represent a kind of observable of an underlying table.

\subsection{Restriction map for relations}
Given an inclusion $U \hookrightarrow S$, we define a summarizing restriction map $\rho_{S \leftarrow U}$ from $\mathbb{S}$-valued relations on $S$ to $\mathbb{S}$-valued relations on $U$ by
\begin{equation} \label{eq_relationrestriction}
\rho_{S \leftarrow U} :=  f \circ \pi_{\id} \circ \tau_{S \leftarrow U} \circ \ell
\end{equation}

What this map does is first lift ($\ell$) the relation to a table by assigning each tuple an index; then restrict with the index distinguishing duplicate tuples on the restricted alphabet; sum over the duplicates to obtain a summary which is again a relation; then forget $f$ the index ($\ell \circ f$ would reset the index).  When $\mathbb{S} = \mathbb{B}$ the Boolean semiring, the summation is just an OR, so this marginalization coincides with the usual notion of restriction of a relation.  When $\mathbb{S} = \N$ or $\R$, this marginalizes the relation.

\section{(Pre)sheaves of relations and tables, morphisms, and contextuality} \label{sec_contextual}

A database consists of tables as in the previous section, each associated to a set of columns or contexts $C$ drawn from variables $X_1, \dots, X_n$.  These sets can overlap and cover their union $[n]$.  This means that as observed by Abramsky, presheaves provide a natural way to think about generalizing a single table or relation to a database of related  tables or relations.  The contexts $C \in \mcM=\{C_1, \dots, C_k\}$ representing a cover of $[n]$ generate a partially ordered set, the open sets in a topological space which sits inside the poset of subsets of $[n]$.  This poset can be interpreted as a category $\mcC$.

Let $\mcC$ be a site (such as this topological space), where the morphisms are inclusion maps, and $\mcV$ be a concrete 
value category such as $\Set$; then an $\mcV$-valued presheaf on $\mcC$ is a contravariant functor $F\!:\mcC \ra \mcV$.

Elements $s \in F(U)$ are called sections, and a family giving such an element for each $U$ is called a family of local sections. An element of $F([n])$ is called a global section.  The maps $F(S \hookrightarrow U)$ are called restriction maps $F(U) \ra F(S)$.

\begin{defn} \label{def_compat_family}
Let $U$ be an open set and $(U_i)_{i \in J}$ a cover; a family $(s_i \in F(U_i))_{i \in J}$ of local sections is a {\em compatible family} if $s_i |_{U_i \cap U_j} = s_j |_{U_i \cap U_j}$ for all $i,j \in J$.  
\end{defn}

A {\em sheaf} is a presheaf satisfying the following two properties.
\begin{enumerate}
\item {\bf Locality:} If $s,t \in F(U)$ agree on every set of an open cover of $U$, $s|_{U_i} = t|_{U_i}$, then they are equal $s=t$.  A presheaf with this property is {\em separated}.
\item {\bf Gluing:}  Given a compatible family $(s_i \in F(U_i))_{i \in J}$, there exists a section $s \in F(U)$ restricting to all the $s_i$.
\end{enumerate}

\begin{defn}
  The adjective {\em contextual} describes a presheaf satisfying Locality but not Gluing, or a particular compatible family of local sections in a presheaf which serves as a counterexample to the Gluing condition.
\end{defn}


\subsection{Presheaves of tables and table-spaces}
There are two levels of presheaves and sheaves we will need, and several flavors (relations, tables, missing data, versions) in each level.  Let $U \subset [n]$.  At the first level, $T(U)$ is a table, and the sections of $T$ are rows; at the second $\scrT(U)$ is the space of all tables on $U$, and the sections of $\scrT$ are $T$s.

\begin{defn}\label{defn_tablepresheaf}
  Fix $n$ and states $\{\Sigma_i\}_{i = 1, \dots, n}$, and a measurement scenario $\mcM$. With respect to the topology generated by $\mcM$, we define a {\em compatible presheaf of tables} as follows.  Assign to each open set $U \subset [n]$ a versioned table with missing data $T(U)$, which is arbitrary except that if $s \in T(U)$, $t \in T(W)$ share the same version and index $i$, we must have $s|_{U \cap W} = t|_{U \cap W}$.  The restriction map  
  $T(U \hookrightarrow S)$ (as in Eq. \eqref{eq_tablerestriction}) is a map from row to row, $(v,i,\sigma_U,s) \mapsto (v,i, \sigma_U|_S,s)$.
\end{defn}

Then a section $s \in T(U)$ is a row or tuple $(v,i,\sigma_U,s)$, and together the tuples make up the table $T(U)$.  
The extra intersection condition in Definition \ref{defn_tablepresheaf} ensures that a family of sections that share a (version and) index $i$ are always a compatible family.

We should check that we have really defined a presheaf, i.e. that $T$ is a functor.  First, $T(U \hookrightarrow U)$ is the identity map on the table $T(U)$, because $\tau_{U\leftarrow U}$ is the identity on $\Sigma_U$.  Second, if $f\:U \hookrightarrow W$ and $g\!:W \hookrightarrow S$ are inclusion maps, $T(g \circ f) = T(g) \circ T(f)$, again because we work pointwise and this holds for state sets.  
  So we have a presheaf of tables.

A global section would be a single row, a single global state $(v,i,\sigma_{[n]}, s)$. 
The extra compatibility hypothesis in Definition \ref{defn_tablepresheaf} is there to prevent sections that share the same index but do not agree when restricted, $\sigma_U|_{U \cap W} \neq \sigma_W |_{U \cap W}$.  This is not quite enough; we want any sections which share an index to be essentially the same section (restrictions of a unique section on some larger set) 
The next sheafification proposition says that such a gluing can always be performed, and that after adding any such glued records, we have a sheaf.

\begin{proposition}
Every presheaf of tables can be completed to a sheaf by adding in the relative global sections obtained by gluing.
\end{proposition}
\begin{proof}
For a presheaf of tables, locality asserts that if two rows in a table $T(U)$ are such that their restrictions to every set of an open cover of $U$ are equal, they are equal.  This is always the case, so a presheaf of tables (and of relations) is always separated.  The gluing axiom says that given an open cover $(U_i)_{i \in J}$ of $U$ and a compatible family of rows $(s_i \in T(U_i))_{i \in J}$, there exists a single row $s \in T(U)$ restricting to all the $s_i$. This row will have the shared $v,i$, and $s$ of the restrictions, and its state will be the gluing $\sigma_U \in \Sigma_U$ of their states. We complete the presheaf by adding in all such glued rows.
\end{proof}

At the next, {\bf table-space} level, we define a functor $\scrT$ so that $\scrT(U)$ is the set of all (versioned, with missing data) tables on columns $U$.  A section is then a particular table, and a global section is a table on all columns $[n]$.  A family of local sections of $\scrT$ is a table $T(U) \in \scrT(U)$ for each $U$. 

\begin{defn} Fix states $\{\Sigma_i\}_{i = 1, \dots, n}$, and a measurement scenario $\mcM$. With respect to the topology generated by $\mcM$, define a {\em presheaf of table-spaces} by assigning to each open set $U \subset [n]$ the set $\scrT(U)$ of all possible tables (versioned and with missing data) on columns $U$, and to each inclusion $U \hookrightarrow S$ the table restriction map  $\tau_{S \leftarrow U}$ of \eqref{eq_tablerestriction}, sending tables $T(U) \in \scrT(U)$ to tables $T(U)|_S$. 
\end{defn}

\begin{proposition}
A presheaf of table-spaces is in fact a presheaf.
\end{proposition}

A compatible presheaf of tables is a more general notion than a compatible family of local sections of $\scrT$, and more appropriate for a model of the state of a distributed data collection system.  In the extreme, a situation such as a network partition can be represented as a compatible presheaf of tables with disjoint version-index sets appearing in tables $T(U)$, $T(W)$ even if $U \cap W \neq \emptyset$, while this is impossible in a compatible family of local sections of $\scrT$ (although permitted in a family of sections of $\scrT$).

\begin{proposition}
Every compatible family of local sections of $\scrT$ for $[n]$ is a compatible presheaf of tables (Definition \ref{defn_tablepresheaf}). 
The converse holds if 
every $T(U_i)$ has the same set of version and index prefixes.
\end{proposition}
\begin{proof}

A compatible family of local sections of $\scrT$ for $[n]$ is a cover $U_i$ of $[n]$ with sections $T(U_i) \in \scrT(U_i)$ such that for all $i,j$, we have  
$T(U_i)|_{U_i \cap U_j} = T(U_j)|_{U_i \cap U_j}$.  Suppose that $s_i \in T(U_i)$ and $s_j \in T(U_j)$ have the same index and version.  The restriction maps preserve version-index sets, and this $(v,i)$ prefix is unique in $T(U_i)$ and also in $T(U_j)$, so the unique $(v,i)$-indexed row in the left restriction comes from $s_i$ and the unique $(v,i)$ indexed row in the right restriction comes from $s_j$.  Since this is a compatible family of local sections of $\scrT$, these rows must be equal for the restrictions to be equal.

On the other hand, suppose we have a compatible presheaf of tables $T$, but $T(U_i)|_{U_i \cap U_j} \neq T(U_j)|_{U_i \cap U_j}$.  Suppose first the two restrictions have the same set of $(v,i)$ prefixes in their constituent tuples.  Then there is some $(v,i)$-prefixed tuple in the left restriction whose state or semigroup element is unequal to a tuple in the right restriction with the same $(v,i)$ prefix.  Since restriction preserves version and index, this contradicts $T$ a compatible presheaf of tables.

Next suppose the two restrictions have unequal sets of $(v,i)$ prefixes; split the restrictions into shared and unique parts based on prefix. The shared parts are equal by the previous argument; the rest can be arbitrary.
\end{proof}

When is a presheaf of table-{\em spaces} a sheaf?  
For a presheaf of table-{\em spaces}, locality asserts that if two tables $T(U),T'(U) \in \scrT(U)$ have $T(U)|_{U_i} = T'(U)|_{U_i}$ in every $U_i$ of an open cover of $U$ are equal, $T(U) = T'(U)$. Because indexing makes equality quite strict, even if the open cover consists of single variables, this holds, so a presheaf of table-spaces is always separated. 
The gluing axiom says that given an open cover $(U_i)_{i \in J}$ of $U$ and a compatible family of tables $(T(U_i) \in \scrT(U_i))_{i \in J}$, there exists a single table $T(U) \in F(U)$ restricting to all the $T(U_i)$. Again since our notion of restriction for tables is index-preserving and this condition is stricter that being a compatible presheaf of tables, such a gluing can always be constructed.  The glued section may not already appear in $\scrT(U)$, so again we can add it to complete a presheaf of table-spaces $\scrT$ to a sheaf.

Thus contextuality will be obtained from presheaf versions of relational summaries of tables.  

\subsection{Presheaves of relation-spaces}

Now we turn to the case of relations.  The uniqueness and lack of indices, versions,  and missing data makes these simpler, but also more complex: the restriction map is no longer pointwise (index-preserving), but necessarily involves summing over indices.  The reference \cite{abramsky2013relational} deals with the case of a family of local sections of a presheaf of semiring-valued relation-spaces.  






\begin{defn} \label{def_presheaf_relation_spaces}
Fix $n$ and states $\{\Sigma_i\}_{i = 1, \dots, n}$, and a measurement scenario $\mcM$. With respect to the topology generated by $\mcM$, we define a {\em presheaf of relation-spaces} $\scrR$ as follows.  Assign to each open set $U \subset [n]$ the set $\scrR(U) = \Hom_{\Set}(\Sigma_U,\mathbb{S})$ of all relations $R: \Sigma_U \ra \mathbb{S}$.  To each inclusion $U \rightarrow S$ associate the relation restriction map $\rho_{S \leftarrow U}$ of \eqref{eq_relationrestriction} as a map from $\Hom_{\Set}(\Sigma_U,\mathbb{S})$ to $\Hom_{\Set}(\Sigma_S,\mathbb{S})$.
\end{defn}

The Bell scenario is a family of compatible local sections of a presheaf of $\R_{\geq 0}$-relation-spaces which has no global section.  An example of a compatible family with a global section is a joint probability distribution on discrete random variables $X_1, \dots, X_n$, together with a collection of marginal distributions defined by the subsets $\mcM$.

\subsection{Sheafy summaries}

Recall the summarization map $\pi_\phi$ of \eqref{eq_summary_table2relation} sending a $\mathbb{S}$-table on $U$ to a $\phi(\mathbb{S})$-relation on $U$, given by $\pi_\phi(T)(\sigma_U) = \sum_{(i,\sigma'_U,s) \in T: \sigma'_U = \sigma_S} \phi \circ T(i,\sigma'_U)$.  
This defines a map $\pi: \scrT \mapsto \scrR$ from a presheaf of table-spaces to a presheaf of relation-spaces by composition, sending $U$ to the table $\scrT(U)$ to the relation $\pi(\scrT(U))$.  Note that beginning with a table with missing data, $\pi \circ \tau \neq \rho \circ \pi$ because the latter map neglects missing data that could be used if restriction was applied first.  
It is interesting to ask if there is a way to modify the treatment of missing data in $\pi$ to make this commute.



\begin{theorem}\label{thm_summaryofMDTableiscontextual}
  When the tables have missing data, a compatible family of local sections which glues to a global section can be sent by $\pi$ to a compatible family of local sections which does not glue.  
\end{theorem}
\begin{proof}
An example is given in Table \ref{table:BellNA}.
\end{proof}

\begin{observation} \label{obs_avail}
  In the {\em available case analysis} (or ``pairwise deletion'') method for dealing with missing data, parameters are estimated based on available data, even if that means the sample size varies. Thus if a set of parameters depends only on the relation of counts for the variables in $U$, the map $\pi_U \circ \tau_{U \leftarrow [n]}$ is applied to obtain it.  Thus what we describe here maps closely to the available case analysis.  
Theorem \ref{thm_summaryofMDTableiscontextual} shows that available case analysis can produce a contextual empirical model.
\end{observation}





In the other direction, we can view the use of tables, missing data, and versions, and even negative probabilities as means to resolve the apparent contradictions in a contextual family of $\mathbb{S}$-relations and obtain a global section. 


\subsection{Versions and $\pi$-compatible concurrent snapshots} \label{sec_vers_conc_snap}
Let us now model the apparent state of the world from the point of view of a network of spatially-distributed computational agents, each of which collects data from its own neighborhood (which overlaps that of others) at a single instant (more precisely, at spacetime events which are spacelike relative to each other so are simultaneous in some frame).  The simultaneous shared state of this network will be represented as a compatible family of local sections of a presheaf of table-spaces. This setup will also be useful in describing distributed systems which share entangled quantum state.

Let $V$ be a version poset (a causet), and $\omega\!: \mcM \ra V$ be a map assigning a version to each maximal measurement context.  Each $C_k \in \mcM$ models the largest set of variables that can be observed at one event (a point in spacetime, or more precisely in our Section \ref{sec:time} model a node in a spacegraph and an interval in the local clock at that node).  On the other hand $V$ represents the causet of events in a patch of spacetime. Hence $\omega$ connects the two, defining an event at which each context is observed ($\omega$ will be injective if the contexts are truly maximal).  When $\omega(\mcM)$ is an antichain in $V$, we have observations which are completely spacelike: no observation can know the outcome of any other and any two can be in write skew.

We would like to define a family of tables representing what these agents see.  Given a compatible presheaf of versioned tables (with or without missing data) $T$ and an open set $U$,  define a family of tables $T^\omega$ by setting $T^\omega(U) = T_{\leq \omega(U)}(U)$ 
so in particular $T^\omega(C_j) = T_{\leq \omega(C_j)}(C_j)$ for each context.

We will need a compatible family of local sections of a table- or relation-space to obtain contextuality.  Thus we want restrictions to agree, $T^\omega(U)|_{U \cap V} = T^\omega(V)|_{U \cap V}$, or at least summaries of these restrictions to agree.  But by definition they can disagree on version number (and always will if $U,V$ are maximal contexts and $\omega$ is injective).

On the other hand, if we ignore the version numbers, it is easy to have a restriction that is unequal if e.g.\ agent $C_{AB}$ and agent $C_{AB'}$ both update the value of $A$ in the row with index $i$ concurrently. 

Thus to make the compatibility condition meaningful, we should apply it to a summary of some kind, such as the summary to $\N$-relations (tables of counts).

Fixing our sheafy summary $\pi$, we want to require that $\pi(T^\omega)$ is a compatible presheaf or a family of compatible sections of a sheaf of relation-spaces (or table spaces): each context can see a different snapshot, but they must agree on any overlaps up to the summary $\pi$.  Generally this will be a relational summary (forgetting index and version).

Given a table $T$ or relation $R$ for every set in a cover (e.g.\ every context), we can attempt to generate a family of tables (relations) for each $U$ by applying restriction maps. This family may be ill defined, because what to assign for $T(U)$ is unclear if $U \subset C_i \cap C_j$ but $T(C_i)|U \neq T(C_j)|U$.  The next definition requires this generated family of relations to be well defined for $\pi T$.

\begin{defn} \label{def_pi_compcon_snap}
  Fix states, contexts $\mcM$, a $V$-versioned presheaf $T$ of tables (possibly with missing data), a map $\omega\:\mcM \ra V$ assigning a version to each maximal context, and a sheafy summary map $\pi$ sending each $T(U)$ to a $\mathbb{S}$-relation. 
  The family $T^\omega$ of $|\mcM|$ tables given by $T^\omega(C_j) = T_{\leq \omega(C_j)}(C_j)$ for each context $C_j$ 
  is a {\em $\pi$-compatible concurrent snapshot} if $\pi(T^\omega)$ generates a compatible family of local sections of a presheaf of relation-spaces.
\end{defn}

Thus a $\pi$-compatible concurrent snapshot represents the ``simoultaneous'' viewpoint of several agents, each responsible for a context, who agree as much as possible (up to $\pi$) while maintaining independent versions.


In Example \ref{example_versionbell}, we see that such a summary $\pi T^\omega$ of a $\pi$-compatible concurrent snapshot can be contextual, even if there is no missing data, $T^\omega$ glues to a global section, and the summary of each local snapshot $T_{\leq \omega(C_j)}(C_j)$ is noncontextual.  Thus contextuality identical to that in Bell's theorem can arise from staleness (write skew) alone.  This is relevant to any algorithm in which distributed agents make decisions knowing that they only have access to the information in their past light cone, but cannot lock, abort \cite{cahill2009serializable} or wait to preclude the possibility of write skew.


\section*{Acknowledgments}
The author gratefully acknowledges the support of AFOSR Grant FA9550-16-1-0300.
\bibliographystyle{custom} 
\bibliography{bibfile}
\end{document}